\documentclass[reqno,10pt,letterpaper]{amsart}

\usepackage{amsmath}
\usepackage{amssymb}
\usepackage{amsfonts}
\usepackage{amsthm}
\usepackage[foot]{amsaddr}
\usepackage{mathtools}
\mathtoolsset{%
}

\usepackage[utf8]{inputenc}
\usepackage[T1]{fontenc}
\usepackage{comment}

\usepackage[sf,mono=false]{libertine}

\usepackage{dsfont}

\usepackage[%
cal=cm,
]
{mathalfa}

\usepackage[dvipsnames,svgnames]{xcolor}
\colorlet{MyBlue}{DodgerBlue!75!Black}
\colorlet{MyGreen}{DarkGreen!85!Black}

\usepackage[font=footnotesize,labelfont=bf]{caption}
\usepackage{subfigure}
\usepackage{tikz}

\usepackage{acronym}
\usepackage{latexsym}
\usepackage{nicefrac}
\usepackage{paralist}
\usepackage{xspace}

\usepackage[numbers,sort&compress]{natbib}

\usepackage{hyperref}
\hypersetup{
colorlinks=true,
linktocpage=true,
pdfstartview=FitH,
breaklinks=true,
pdfpagemode=UseNone,
pageanchor=true,
pdfpagemode=UseOutlines,
plainpages=false,
bookmarksnumbered,
bookmarksopen=false,
bookmarksopenlevel=1,
hypertexnames=true,
pdfhighlight=/O,
urlcolor=Maroon,linkcolor=MyBlue!60!black,citecolor=DarkGreen!70!black, 
pdftitle={},
pdfauthor={},
pdfsubject={},
pdfkeywords={},
pdfcreator={pdfLaTeX},
pdfproducer={LaTeX with hyperref}
}

\newcommand{\EMAIL}[1]{\email{\href{mailto:#1}{#1}}}

\numberwithin{equation}{section}  
\usepackage[sort&compress,capitalize,nameinlink]{cleveref}

\newcommand{\dd}{\:d}

\newcommand{\eps}{\varepsilon}
\newcommand{\from}{\colon}

\newcommand{\pd}{\partial}

\newcommand{\bP}{\mathbf{P}}

\newcommand{\boldb}{\mathbf{b}}

\newcommand{\bx}{\mathbf{x}}
\newcommand{\by}{\mathbf{y}}

\newcommand{\R}{\mathbb{R}}

\newcommand{\N}{\mathbb{N}}

\DeclareMathOperator*{\argmax}{arg\,max}
\DeclareMathOperator*{\argmin}{arg\,min}

\DeclareMathOperator{\bigoh}{\mathcal O}

\DeclareMathOperator{\divg}{div}

\DeclareMathOperator{\ex}{\mathbb{E}}

\DeclareMathOperator{\prob}{\mathbb{P}}

\DeclareMathOperator{\supp}{supp}

\DeclarePairedDelimiter{\braces}{\{}{\}}
\DeclarePairedDelimiter{\bracks}{[}{]}

\DeclarePairedDelimiter{\abs}{\lvert}{\rvert}
\DeclarePairedDelimiter{\norm}{\lVert}{\rVert}
\DeclarePairedDelimiterXPP{\dnorm}[1]{}{\lVert}{\rVert}{_{\ast}}{#1}

\DeclarePairedDelimiterX{\braket}[2]{\langle}{\rangle}{#1,#2}

\DeclarePairedDelimiterX{\inner}[2]{\langle}{\rangle}{#1,#2}
\DeclarePairedDelimiterX{\setdef}[2]{\{}{\}}{#1:#2}

\DeclarePairedDelimiterXPP{\exclude}[1]{\mathopen{}\setminus}{\{}{\}}{}{#1}

\DeclarePairedDelimiterXPP{\probof}[1]{\prob}{(}{)}{}{%

#1}

\DeclarePairedDelimiterXPP{\exof}[1]{\ex}{[}{]}{}{%

#1}

\newcommand{\dis}{\displaystyle}
\newcommand{\txs}{\textstyle}
\newcommand{\textpar}[1]{\textup(#1\textup)}

\newcommand{\insum}{\sum\nolimits}

\usepackage[textwidth=30mm,disable]{todonotes}

\usepackage{algorithm,algorithmic}

\theoremstyle{plain}
\newtheorem{theorem}{Theorem}

\newtheorem*{corollary*}{Corollary}
\newtheorem{lemma}[theorem]{Lemma}

\theoremstyle{definition}
\newtheorem{definition}[theorem]{Definition}
\newtheorem*{definition*}{Definition}

\theoremstyle{remark}

\newtheorem*{remark*}{Remark}
\newtheorem{example}{Example}

\numberwithin{equation}{section}
\numberwithin{theorem}{section}
\numberwithin{remark}{section}
\numberwithin{example}{section}

\newcommand{\bvec}{e}

\newcommand{\feas}{\mathcal{X}}
\newcommand{\intfeas}{\feas^{\circ}}

\newcommand{\base}{p}

\newcommand{\sol}{x^{\ast}}

\DeclareMathOperator{\Eucl}{\Pi}
\DeclareMathOperator{\logit}{\Lambda}

\newcommand{\choice}{Q}
\newcommand{\fench}{G}
\newcommand{\lyap}{L}

\newcommand{\strong}{K}
\newcommand{\depth}{\Omega}

\newcommand{\flow}{\Phi}
\newcommand{\run}{n}

\newcommand{\temp}{\eta}

\newcommand{\play}{i}
\newcommand{\playalt}{j}
\newcommand{\nPlayers}{N}
\newcommand{\players}{\mathcal{\nPlayers}}

\newcommand{\pure}{\alpha}
\newcommand{\purealt}{\beta}
\newcommand{\nPures}{A}
\newcommand{\pures}{\mathcal{\nPures}}

\newcommand{\strat}{x}

\newcommand{\strats}{\feas}
\newcommand{\intstrats}{\intfeas}

\newcommand{\pay}{u}
\newcommand{\payv}{v}

\newcommand{\game}{\Gamma}
\newcommand{\gamefull}{\game(\players,\pures,\pay)}

\newcommand{\eq}{\strat^{\ast}}
\newcommand{\peq}{\pure^{\ast}}
\newcommand{\ptest}{\hat\pure}
\newcommand{\olim}{\hat\strat}

\DeclareMathOperator{\reg}{Reg}

\newcommand{\graph}{\mathcal{G}}
\newcommand{\nodes}{\players}
\newcommand{\edges}{\mathcal{E}}

\newcommand{\edge}{e}

\newcommand{\ball}{\mathbb{B}}

\newcommand{\intsimplex}{\simplex^{\!\circ}}

\newcommand{\simplex}{\Delta}

\newcommand{\dynfield}{V}

\newcommand{\PM}[1]{\todo[color=DodgerBlue!20!LightGray,author=\textbf{PM},inline]{\small #1\\}}

\newcommand{\CP}[1]{\todo[color=DarkGreen!20!LightGray,author=\textbf{CP},inline]{\small #1\\}}

\allowdisplaybreaks
\begin{document}


\title{Cycles in adversarial regularized learning}

\author{Panayotis Mertikopoulos$^{\ast}$}
\address{$^{\ast}$ Univ. Grenoble Alpes, CNRS, Inria, LIG, F-38000 Grenoble, France.}
\EMAIL{panayotis.mertikopoulos@imag.fr}

\author{\\Christos Papadimitriou$^{\S}$}
\address{$^{\S}$ UC Berkeley.}
\EMAIL{christos@berkeley.edu}

\author{Georgios Piliouras$^{\ddag}$}
\address{$^{\ddag}$ Singapore University of Technology and Design.}
\EMAIL{georgios@sutd.edu.sg}

\thanks{%
Panayotis Mertikopoulos was partially supported by
the French National Research Agency (ANR) project ORACLESS (ANR\textendash GAGA\textendash13\textendash JS01\textendash 0004\textendash 01)
and the Huawei Innovation Research Program ULTRON.
Georgios Piliouras would like to acknowledge SUTD grant SRG ESD 2015 097 and MOE AcRF Tier 2 Grant  2016-T2-1-170.
}



\newcommand{\acdef}[1]{\textit{\acl{#1}} \textup{(\acs{#1})}\acused{#1}}
\newcommand{\acdefp}[1]{\emph{\aclp{#1}} \textup(\acsp{#1}\textup)\acused{#1}}

\newacro{FRL}[FoReL]{``Follow the Regularized Leader''}
\newacro{MW}{multiplicative weights}
\newacro{MWU}{multiplicative weights update}
\newacro{LHS}{left-hand side}
\newacro{RHS}{right-hand side}
\newacro{NE}{Nash equilibrium}
\newacroplural{NE}[NE]{Nash equilibria}
\newacro{CCE}{coarse correlated equilibrium}
\newacroplural{CCE}[CCE]{coarse correlated equilibria}
\newacro{SA}{stochastic approximation}
\newacro{iid}[i.i.d.]{independent and identically distributed}

\begin{abstract}
%
%
Regularized learning is a fundamental technique in online optimization, machine learning and many other fields of computer science. 
A natural question that arises in these settings is how regularized learning algorithms behave when faced against each other.
 We study a natural formulation of this problem by coupling regularized learning dynamics in zero-sum games. We show that the system's behavior is Poincar\'{e} recurrent, implying that almost every trajectory revisits any (arbitrarily small) neighborhood of its starting point infinitely often. This cycling behavior is robust to the agents' choice of regularization mechanism (each agent could be using a different regularizer), to positive-affine transformations of the agents' utilities, and it also persists in the case of networked competition, i.e., for  zero-sum polymatrix games.
\end{abstract}

\maketitle


\acresetall

\section{Introduction}
\label{sec:introduction}

\CP{Regularization is optimization's latest and most incisive twist, its present zeitgeist: 
Through the introduction of a new component to the objective, 
a new algorithm results which overcomes ill-conditioning and overfitting, and
achieves sparsity and parsimony without sacrificing efficiency.
In the context of on-line optimization, regularization is exemplified...}

Regularization is a fundamental and incisive method in optimization, its present \emph{zeitgeist} and its entry into machine learning.
Through the introduction of a new component in the objective,
regularization techniques overcome ill-conditioning and overfitting, and they yield algorithms that achieve sparsity and parsimony without sacrificing efficiency \cite{ben2001lectures,Cesa06,Arora05themultiplicative}.

In the context of online optimization, these features are exemplified in the family of learning algorithms known as \acdef{FRL} \cite{SSS07}.
\ac{FRL} represents an important archetype of adaptive behavior for several reasons:
it provides optimal min-max regret guarantees ($\bigoh(t^{-1/2})$ in an adversarial setting),
it offers significant flexibility with respect to the geometry of the problem at hand,
and
it captures numerous other dynamics as special cases (hedge, \acl{MW}, gradient descent, etc.)
\cite{hazan2016introduction,Cesa06,Arora05themultiplicative}. 
As such, given that these regret guarantees hold without any further assumptions about how payoffs/costs are determined at each stage, the dynamics of \ac{FRL} have been the object of intense scrutiny and study in algorithmic game theory.

The standard way of analyzing such no-regret dynamics in games involves a two-step approach.
The first step exploits the fact that the empirical frequency of play under a no-regret algorithm converges to the game's set of \acfp{CCE}.
The second involves proving some useful property of the game's \acp{CCE}:
For instance, leveraging $(\lambda, \mu)$-robustness \cite{Roughgarden09} implies that the social welfare at a \ac{CCE} lies within a small constant of the optimum social welfare;
as another example, the product of the marginal distributions of \ac{CCE} in zero-sum games is Nash.
In this way, the no-regret properties of \ac{FRL} can be turned into convergence guarantees for the players' empirical frequency of play (that is, in a time-averaged, correlated sense).

Recently, several papers have moved beyond this ``black-box'' framework and focused instead on obtaining stronger regret/convergence guarantees for systems of learning algorithms coupled together in games with a specific structure.
Along these lines, Daskalakis et al. \cite{Daskalakis:2011:NNA:2133036.2133057} and Rakhlin and Sridharan  \cite{rakhlin2013optimization} developed classes of dynamics that enjoy a $\bigoh(\log t/t)$ regret minimization rate in two-player zero-sum games.
Syrgkanis et al. \cite{Syrgkanis:2015:FCR:2969442.2969573} further analyzed a recency biased variant of \ac{FRL} in more general multi-player games and showed that it is possible to achieve an $\bigoh(t^{-3/4})$ regret minimization rate.
The social welfare converges at a rate of $\bigoh(t^{-1})$, a result which was extended to standard versions of \ac{FRL} dynamics in \cite{foster2016learning}.

Whilst a regret-based analysis provides significant insights about these systems,
it does not answer a fundamental behavioral question:
\begin{quote}
\centering
\itshape
Does the system converge to a \acl{NE}?
\\
Does it even stabilize?
\end{quote}
The dichotomy between a self-stabilizing, convergent system and a system with recurrent cycles is of obvious significance, but a regret-based analysis cannot distinguish between the two.
Indeed, convergent, recurrent, and even chaotic \cite{2017arXiv170301138P} systems may exhibit equally strong regret minimization properties in general games, so the question remains:
What does the long-run behavior of \ac{FRL} look like, really?

This question becomes particularly interesting and important under perfect competition (such as zero-sum games and variants thereof).
Especially in practice, zero-sum games can capture optimization ``duels'' \cite{immorlica2011dueling}:
for example, two Internet search engines competing to maximize their market share can be modeled as players in a zero-sum game with a convex strategy space.
In \cite{immorlica2011dueling} it was shown that the time-average of a regret-minimizing class of dynamics converges to an approximate equilibrium of the game. 
Finally, zero-sum games have also been used quite recently as a model for deep learning optimization techniques in image generation and discrimination \cite{goodfellow2014generative,schuurmans2016deep}.

In each of the above cases, min-max strategies are typically thought of as the axiomatically correct prediction.
The fact that the time average of the marginals of a \ac{FRL} procedure converges to such states is considered as further evidence of the correctness of this prediction.
However, the long-run behavior of the \emph{actual} sequence of play (as opposed to its time-averages) seems to be trickier, and a number of natural questions arise:
\begin{itemize}
[-]
\itshape
\item
Does optimization-driven learning converge under perfect competition?
\item
Does fast regret minimization necessarily imply \textpar{fast} equilibration in this case?
\end{itemize}


\subsection*{Our results}
We settle these questions with a resounding ``no''.
Specifically, we show that the behavior of \ac{FRL} in zero-sum games with an interior equilibrium (e.g. Matching Pennies) is \emph{Poincaré recurrent}, implying that almost every trajectory revisits any (arbitrarily small) neighborhood of its starting point infinitely often.
Importantly, the observed cycling behavior is robust to the agents' choice of regularization mechanism (each agent could be using a different regularizer), and it applies to any positive affine transformation of zero-sum games (and hence all strictly competitive games \cite{adler2009note}) even though these transformations lead to \emph{different} trajectories of play.
Finally, this cycling behavior also persists in the case of networked competition, i.e. for constant-sum polymatrix games   \cite{cai2016zero,Cai,DP09}.


\begin{figure}[tbp]
\centering
\subfigure{\includegraphics[width=.44\textwidth]{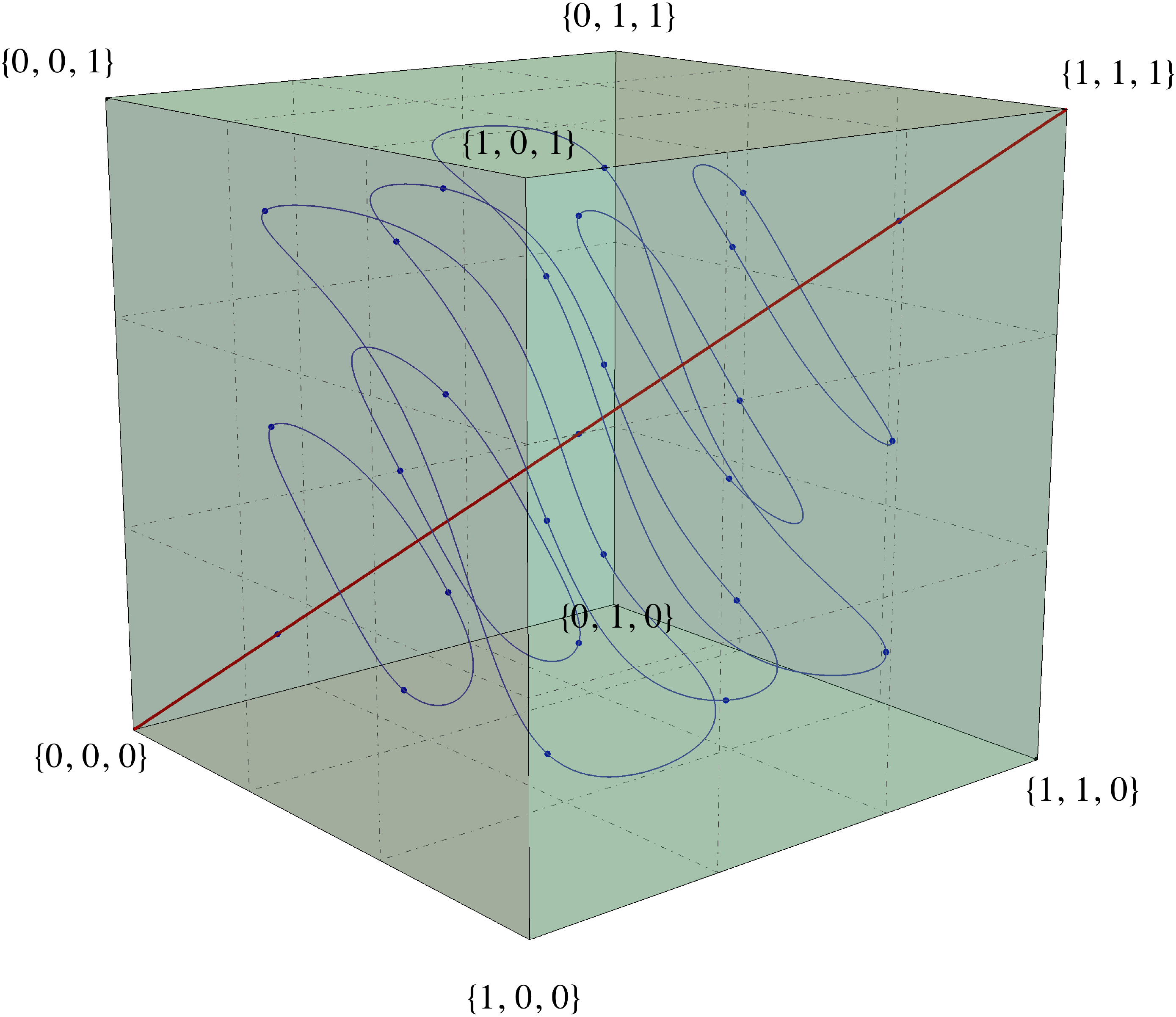}}
\hfill
\subfigure{\includegraphics[width=.44\textwidth]{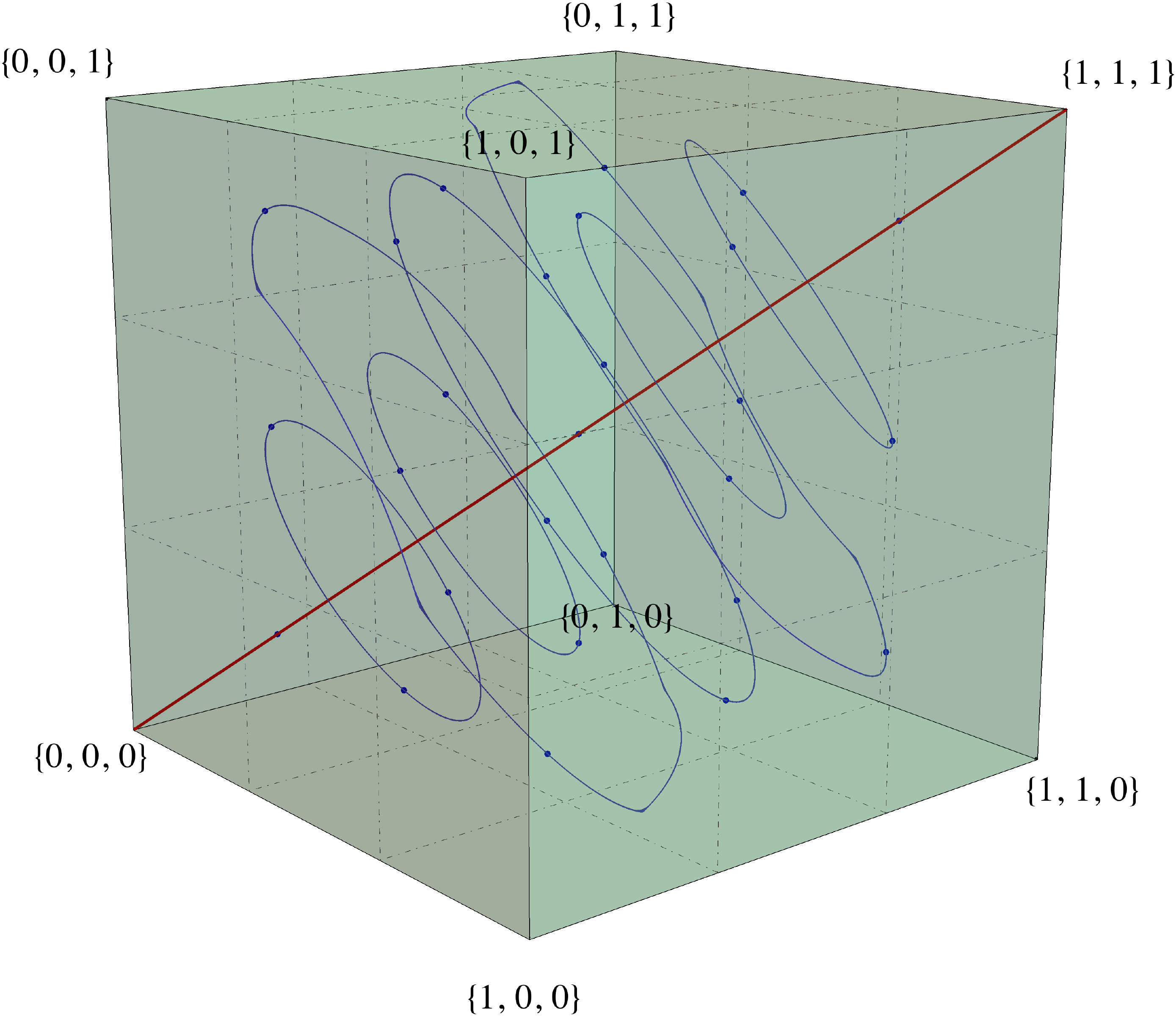}}%
\caption{Evolution of the dynamics of \ac{FRL} in a $3$-player zero-sum polymatrix game with entropic and Euclidean regularization (left and right respectively).
The game considered is a graphical variant of Matching Pennies with three players.
As can be seen, the trajectories of \ac{FRL} orbit the game's line of \aclp{NE} (dark red).
The kinks observed in the Euclidean case occur when the support of the trajectory of play changes;
by contrast, the \acl{MW} dynamics (left) are interior, so they do not exhibit such kinks.}%
\vspace{-1ex}
\label{fig:recurrence}
\end{figure}


Given that the no-regret guarantees of \ac{FRL} require a decreasing step-size (or learning rate),%
\footnote{A standard trick is to decrease step-sizes by a constant factor after a window of ``doubling'' length \cite{SS11}.}
we focus on a smooth version of \ac{FRL} described by a dynamical system in continuous time.
The resulting \ac{FRL} dynamics enjoy a particularly strong $\bigoh(t^{-1})$ regret minimization rate and they capture as a special case the replicator dynamics \cite{TJ78,Tay79,SS83} and the projection dynamics \cite{Fri91,SDL08,MS16}, arguably the most widely studied game dynamics in biology, evolutionary game theory and transportation science \cite{Hofbauer98,Weibull,San10}.
In this way, our analysis unifies and generalizes many prior results on the cycling behavior of evolutionary dynamics \cite{Hofbauer98,piliouras2014optimization,PiliourasAAMAS2014,Sato02042002} and it provides a new interpretation of these results through the lens of optimization and machine learning.

From a technical point of view, our analysis touches on several issues.
Our first insight is to focus not on the simplex of the players' mixed strategies, but on a \emph{dual} space of payoff differences.
The reason for this is that the vector of cumulative payoff differences between two strategies fully determines a player's mixed strategy under \ac{FRL}, and it is precisely these differences that ultimately drive the players' learning process.
Under this transformation, \ac{FRL} exhibits a striking property, \emph{incompressibility}:
the flow of the dynamics is volume-preserving, so a ball of initial conditions in this dual space can never collapse to a point.

That being said, the evolution of such a ball in the space of payoffs could be \emph{transient}, implying in particular that the players' mixed strategies could converge (because the choice map that links payoff differences to strategies is nonlinear).
To rule out such behaviors, we show that \ac{FRL} in zero-sum games with an interior \acl{NE} has a further important property:
it admits a \emph{constant of motion}.
Specifically, if $\eq = (\eq_{\play})_{\play\in\players}$ is an interior equilibrium of the game and $y_{\play}$ is an arbitrary point in the payoff space of player $\play$, this constant is given by the coupling function
\begin{equation}
\notag
\fench(y)
	= \sum_{\play\in\players} \bracks{h_{\play}^{\ast}(y_{\play}) - \braket{y_{\play}}{\eq_{\play}}},
\end{equation}
where $h_{\play}^{\ast}(y_{\play}) = \max_{x_{\play}} \{\braket{y_{\play}}{x_{\play}} - h_{\play}(x_{\play})\}$ is the convex conjugate of the regularizer $h_{\play}$ that generates the learning process of player $\play$ (for the details, see \cref{sec:learning,sec:results}).
Coupled with the dynamics' incompressibility, this invariance can be used to show that \ac{FRL} is \emph{recurrent}:
after some finite time, almost every trajectory returns arbitrarily close to its initial state.

On the other hand, if the game does not admit an interior equilibrium, the coupling above is no longer a constant of motion.
In this case, $\fench$ decreases over time until the support of the players' mixed strategies matches that of a \acl{NE} with maximal support:
as this point in time is approached, $\fench$ essentially becomes a constant.
Thus, in general zero-sum games, \ac{FRL} wanders perpetually in the smallest face of the game's strategy space containing all of the game's equilibria;
indeed, the only possibility that \ac{FRL} converges is if the game admits a unique \acl{NE} in pure strategies \textendash\ a fairly restrictive requirement.

\section{Definitions from game theory}
\label{sec:prelims}

\subsection{Games in normal form}
\label{sec:games}

We begin with some basic definitions from game theory.
A \emph{finite game in normal form} consists of a finite set of \emph{players} $\players = \{1,\dotsc,\nPlayers\}$,
each with a finite set of \emph{actions} (or \emph{strategies}) $\pures_{\play}$.
The preferences of player $\play$ for one action over another are determined by an associated \emph{payoff function} $\pay_{\play}\from \pures\equiv \prod_{\play} \pures_{\play}\to\R$ which assigns a reward $\pay_{\play}(\pure_{\play};\pure_{-\play})$ to player $\play\in\players$ under the \emph{strategy profile} $(\pure_{\play};\pure_{-\play})$ of all players' actions.%
\footnote{In the above, we use the standard shorthand $(\purealt_{\play};\pure_{-\play})$ for the profile $(\pure_{1},\dotsc,\purealt_{\play},\dotsc,\pure_{\nPlayers})$.}
Putting all this together, a game in normal form will be written as a tuple $\game \equiv \gamefull$ with players, actions and payoffs defined as above.

Players can also use \emph{mixed strategies}, i.e. mixed probability distributions $x_{\play} = (x_{\play\pure_{\play}})_{\pure_{\play}\in \pures_{\play}} \in \simplex(\pures_{\play})$ over their action sets $\pures_{\play}$.
The resulting probability vector $x_{\play}$ is called a \emph{mixed strategy} and we write $\strats_{\play} = \simplex(\pures_{\play})$ for the mixed strategy space of player $\play$.
Aggregating over players, we also write $\strats = \prod_{\play} \strats_{\play}$ for the game's \emph{strategy space}, i.e. the space of all strategy profiles $x=(x_{\play})_{\play\in\players}$.

In this context (and in a slight abuse of notation), the expected payoff of the $\play$-th player in the profile $
x = (x_{1},\dotsc,x_{\nPlayers})$ is
\begin{equation}
\label{eq:pay-avg}
\pay_{\play}(x)
	= \sum_{\mathclap{\pure_{1}\in\pures_{1}}} \:\dotsi\: \sum_{\mathclap{\pure_{\nPlayers}\in\pures_{\nPlayers}}} \;
	\pay_{\play}(\pure_{1},\dotsc,\pure_{\nPlayers})\,
	x_{1\pure_{1}} \dotsm x_{\nPlayers\pure_{\nPlayers}}.
\end{equation}
To keep track of the payoff of each pure strategy, we also write $\payv_{\play\pure_{\play}}(x) = \pay_{\play}(\pure_{\play};x_{-\play})$ for the payoff of strategy $\pure_{\play}\in\pures_{\play}$ under the profile $x\in\strats$ and $\payv_{\play}(x) = (\payv_{\play\pure_{\play}}(x))_{\pure_{\play}\in\pures_{\play}}$ for the resulting \emph{payoff vector} of player $\play$.
We then have
\begin{equation}
\label{eq:pay-payv}
\pay_{\play}(x)
	= \braket{\payv_{\play}(x)}{x_{\play}}
	= \sum_{\pure_{\play}\in\pures_{\play}} x_{\play\pure_{\play}} \payv_{\play\pure_{\play}}(x),
\end{equation}
where $\braket{\payv}{x} \equiv \payv^{\top} x$ denotes the ordinary pairing between $\payv$ and $x$.

The most widely used solution concept in game theory is that of a \acdef{NE}, defined here as a mixed strategy profile $\eq\in\strats$ such that
\begin{equation}
\label{eq:Nash}
\tag{NE}
\pay_{\play}(\eq_{\play};\eq_{-\play})
	\geq \pay_{\play}(x_{\play};\eq_{-\play})
\end{equation}
for every deviation $x_{\play}\in\strats_{\play}$ of player $\play$ and all $\play\in\players$.
Writing $\supp(\eq_{\play}) = \setdef{\pure_{\play}\in\pures_{\play}}{\eq_{\play}>0}$ for the support of $\eq_{\play}\in\strats_{\play}$, a \acl{NE} $\eq\in\strats$ is called \emph{pure} if $\supp(\eq_{\play}) = \{\peq_{\play}\}$ for some $\peq_{\play}\in\pures_{\play}$ and all $\play\in\players$.
At the other end of the spectrum, $\eq$ is said to be \emph{interior} (or \emph{fully mixed}) if $\supp(\eq_{\play}) = \pures_{\play}$ for all $\play\in\players$.
Finally, a \acdef{CCE} is a distribution $\pi$ over the set of action profiles  $\pures\equiv \prod_{\play} \pures_{\play}$ such that, for every player $\play\in\players$ and every action $\purealt_{\play}\in \pures_{\play}$, we have
$
		\sum_{\pure \in \pures} \payv_{\play}(\pure) \pi(\pure)
		\geq \sum_{\pure_{-\play} \in \pures_{-\play}} \payv_{\play}(\purealt_\play, \pure_{-\play}) \pi_{\play}(\pure_{-\play})
$, 
where \( \pi_\play(\pure_{-\play}) = \sum_{\pure_\play \in \pure_\play} \pi(\pure_\play, \pure_{-\play}) \) is the marginal distribution of $\pi$ with respect to $\play$.

\subsection{Zero-sum games and zero-sum polymatrix games}
\label{sec:0sum}

Perhaps the most widely studied class of finite games (and certainly the first to be considered) is that of \emph{$2$-player zero-sum games}, i.e. when $\players = \{1,2\}$ and $\pay_{1} = - \pay_{2}$.
Letting $\pay \equiv \pay_{1} = -\pay_{2}$, the \emph{value} of a $2$-player zero-sum game $\game$ is defined as
\begin{equation}
\label{eq:value}
\pay_{\game}
	= \max_{x_{1}\in\strats_{1}} \min_{x_{2}\in\strats_{2}} \pay(x_{1},x_{2})
	= \min_{x_{2}\in\strats_{2}} \max_{x_{1}\in\strats_{1}} \pay(x_{1},x_{2}),
\end{equation}
with equality following from von Neumann's celebrated min-max theorem \cite{vN28}. 
As is well known, the solutions of this saddle-point problem form a closed, convex set consisting precisely of the game's \aclp{NE};
moreover, the players' equilibrium payoffs are simply $\pay_{\game}$ and $-\pay_{\game}$ respectively.
As a result, \acl{NE} is the standard game-theoretic prediction in such games.

An important question that arises here is whether the straightforward equilibrium structure of zero-sum games extends to the case of a \emph{network} of competitors.
Following \cite{DP09,Cai,cai2016zero}, an \emph{$\nPlayers$-player pairwise zero-/constant-sum polymatrix game} consists of an (undirected) \emph{interaction graph} $\graph\equiv\graph(\nodes,\edges)$ whose set of nodes $\nodes$ represents the competing players,
with two nodes $\play,\playalt\in\nodes$ connected by an edge $\edge=(\play,\playalt)$ in $\edges$ if and only if the corresponding players compete with each other in a two-player zero-/constant-sum game.

To formalize this, we assume that
\begin{inparaenum}
[\itshape a\upshape)]
\item
every player has a finite set of actions $\pures_{\play}$ (as before);
and
\item
to each edge $\edge=\{\play,\playalt\} \in \edges$ is associated a two-player game zero-/constant-sum $\game_{\edge}$ with player set $\players_{\edge} = \{\play,\playalt\}$, action sets $\pures_{\play}$ and $\pures_{\playalt}$, and payoff functions $\pay_{\play\playalt} = \gamma_{\{\play,\playalt\}}-\pay_{\playalt\play} \from \pures_{\play}\times\pures_{\playalt}\to\R$ respectively.%
\footnote{In a zero-sum game, we have $\gamma_{\{\play,\playalt\}}=0$ by default.
Since the underlying interaction graph is assumed undirected, we also assume that the labeling of the players' payoff functions is symmetric.
At the expense of concision, our analysis extends to directed graphs, but we stick with the undirected case for clarity.}
\end{inparaenum}
The space of mixed strategies of player $\play$ is again $\strats_{\play} = \simplex(\pures_{\play})$, but the player's payoff is now determined by aggregating over all games involving player $\play$, i.e.
\begin{equation}
\label{eq:pay-tot}
\pay_{\play}(x)
	= \sum_{\playalt\in\nodes_{\play}} \pay_{\play\playalt}(x_{\play},x_{\playalt}),
\end{equation}
where $\nodes_{\play} = \setdef{\playalt\in\nodes}{\{\play,\playalt\}\in\edges}$ denotes the set of ``neighbors'' of player $\play$.
In other words, the payoff to player $\play$ is simply the the sum of all payoffs in the zero-/constant-sum games that player $\play$ plays with their neighbors.

In what follows, we will also consider games which are payoff-equivalent to positive-affine transformations of pairwise constant-sum polymatrix games.
Formally, we will allow for games $\game$ such that there exists a pairwise constant-sum polymatrix game $\game'$ and constants $a_{\play}>0$ and $b_{\play}\in\R$ for each player $\play$  such that $\pay_{\play}^{\game}(\pure) = a_i \pay_{\play}^{\game'}(\pure) + b_{\play}$ for each outcome $\pure\in \pures$.

\section{No-regret learning via regularization}
\label{sec:learning}

Throughout this paper, our focus will be on repeated decision making in low-information environments where the players don't know the rules of the game (perhaps not even that they are playing a game).
In this case, even if the game admits a unique \acl{NE}, it is not reasonable to assume that players are able to pre-compute their component of an equilibrium strategy \textendash\ let alone assume that all players are fully rational, that there is common knowledge of rationality, etc.

With this in mind, we only make the bare-bones assumption that every player seeks to at least minimize their ``regret'', i.e. the average payoff difference between a player's mixed strategy at time $t\geq0$ and the player's best possible strategy in hindsight.
Formally, assuming that play evolves in continuous time, the \emph{regret} of player $\play$ along the sequence of play $x(t)$ is defined as
\begin{equation}
\label{eq:regret}
\reg_{\play}(t)
	= \max_{\base_{\play}\in\strats_{\play}}
		\frac{1}{t}\int_{0}^{t} \bracks{\pay_{\play}(\base_{\play};x_{-\play}(s)) - \pay_{\play}(x(s))} \dd s,
\end{equation}
and we say that player $\play$ has \emph{no regret} under $x(t)$ if $\limsup_{t\to\infty}\reg_{\play}(t) \leq 0$.

The most widely used scheme to achieve this worst-case guarantee is known as \acdef{FRL}, an ex\-plo\-itation-ex\-plo\-ration class of policies that consists of playing a mixed strategy that maximizes the player's expected cumulative payoff (the exploitation part) minus a regularization term (exploration).
In our continuous-time framework, this is described by the learning dynamics
\begin{equation}
\label{eq:FRL}
\tag{FoReL}
\begin{aligned}
y_{\play}(t)
	&= y_{\play}(0) + \int_{0}^{t} \payv_{\play}(
	x(s)) \dd s,
	\\
x_{\play}(t)
	&= \choice_{\play}(y_{\play}(t)),
\end{aligned}
\end{equation}
where the so-called \emph{choice map} $\choice_{\play}\from\R^{\pures_{\play}}\to\strats_{\play}$ is defined as
\begin{equation}
\label{eq:choice}
\choice_{\play}(y_{\play})
	= \argmax_{x_{\play}\in\strats_{\play}} \{\braket{y_{\play}}{x_{\play}} - h_{\play}(x_{\play})\}.
\end{equation}

In the above, the \emph{regularizer function} $h_{\play}\from\strats_{\play}\to\R$ is a convex penalty term which smoothens the ``hard'' $\argmax$ correspondence $y_{\play}\mapsto\argmax_{x_{\play}\in\strats_{\play}} \braket{y_{\play}}{x_{\play}}$ that maximizes the player's cumulative payoff over $[0,t]$.
As a result, the ``regularized leader'' $\choice_{\play}(y_{\play}) =  \argmax_{x_{\play}\in\strats_{\play}} \{\braket{y_{\play}}{x_{\play}} - h_{\play}(x_{\play})\}$ is biased towards the \emph{prox-center} $\base_{\play} = \argmin_{x_{\play}\in\strats_{\play}} h_{\play}(x_{\play})$ of $\strats_{\play}$.
For most common regularizers, the prox-center is interior (and usually coincides with the barycenter of $\strats$), so the regularization in \eqref{eq:choice} encourages exploration by favoring mixed strategies with full support.

In \cref{app:examples}, we present in detail two of the prototypical examples of \eqref{eq:FRL}:
\begin{inparaenum}
[\itshape i\hspace*{1pt}\upshape)]
\item
the \acdef{MW} dynamics induced by the entropic regularizer function $h_{\play}(x) = \sum_{\pure_{\play}\in\pures_{\play}} x_{\play\pure_{\play}} \log x_{\play\pure_{\play}}$ (which lead to the replicator dynamics of evolutionary game theory);
and
\item
the projection dynamics induced by the Euclidean regularizer $h_{\play}(x) = \frac{1}{2} \norm{x_{\play}}^{2}$.
\end{inparaenum}
For concreteness, we will assume in what follows that the regularizer of every player $\play\in\players$ satisfies the following minimal requirements:
\begin{enumerate}
\item
$h_{\play}$ is continuous and strictly convex on $\strats_{\play}$.
\item
$h_{\play}$ is smooth on the relative interior of every face of $\strats_{\play}$ (including $\strats_{\play}$ itself).
\end{enumerate}
Under these basic assumptions, the ``regularized leader'' $\choice_{\play}(y_{\play})$ is well-defined in the sense that \eqref{eq:choice} admits a unique solution.
More importantly, we have the following no-regret guarantee:

\begin{theorem}
\label{thm:no-reg}
A player following \eqref{eq:FRL} enjoys an $\bigoh(1/t)$ regret bound, no matter what other players do.
Specifically, if player $\play\in\players$ follows \eqref{eq:FRL}, then, for every continuous trajectory of play $x_{-\play}(t)$ of the opponents of player $\play$, we have
\begin{equation}
\label{eq:no-reg}
\reg_{\play}(t)
	\leq \frac{\depth_{\play}}{t},
\end{equation}
where
$\depth_{\play} = \max h_{\play} - \min h_{\play}$
is a positive constant.
\end{theorem}

\PM{Guys, review the stuff below to make sure I'm not committing any strategic blunders.}

To streamline our discussion, we relegate the proof of \cref{thm:no-reg} to \cref{app:proofs};
we also refer to \cite{KM17} for a similar regret bound for \eqref{eq:FRL} in the context of online convex optimization.
Instead of discussing the proof, we close this section by noting that \eqref{eq:no-reg} represents a striking improvement over the $\Theta(t^{-1/2})$ worst-case bound for \ac{FRL} in discrete time \cite{SS11}.
In view of this, the continuous-time framework we consider here can be seen as particularly amenable to learning because it allows players seek to minimize their regret (and thus converge to \aclp{CCE}) at the fastest possible rate.

%
%

\section{Recurrence in adversarial regularized learning}
\label{sec:results}

In this section, our aim is to take a closer look at the ramifications of fast regret minimization under \eqref{eq:FRL} beyond convergence to the set of \aclp{CCE}.
Indeed, as is well known, this set is fairly large and may contain thoroughly non-rationalizable strategies:
for instance, Viossat and Zapechelnyuk \cite{VZ13} recently showed that a \acl{CCE} could assign positive selection probability \emph{only} to strictly dominated strategies.
Moreover, the time-averaging that is inherent in the definition of the players' regret leaves open the possibility of complex day-to-day behavior e.g. periodicity, recurrence, limit cycles or chaos \cite{Sato02042002,piliouras2014optimization,2017arXiv170301138P,CRS16}.
Motivated by this, we examine the long-run behavior of the \eqref{eq:FRL} in the popular setting of zero-sum games (with or without interior equilibria) and several extensions thereof.

\PM{George, please doublecheck what I wrote below: it was written in a hurry, so maybe I missed something.
Also, what's your go-to reference for all this? Arnold? Hirsch and Smale?}

A key notion in our analysis is that of \emph{\textpar{Poincaré} recurrence}.
Intuitively, a dynamical system is recurrent if, after a sufficiently long (but \emph{finite}) time, almost every state returns arbitrarily close to the system's initial state.%
\footnote{Here, ``almost'' means that the set of such states has full Lebesgue measure.}
More formally, given a dynamical system on $\strats$ that is defined by means of a \emph{semiflow} $\flow\from\strats\times[0,\infty)\to\strats$,
we have:%
\footnote{Recall that a continuous map $\flow\from\strats\times[0,\infty)\to\strats$ is a \emph{semiflow} if $\flow(x,0) = x$ and $\flow(x,t+s) = \flow(\flow(x,t),s)$ for all $x\in\strats$ and all $s,t\geq0$.
Heuristically, $\flow_{t}(x) \equiv \flow(x,t)$ describes the trajectory of the dynamical system starting at $x$.}

\begin{definition}
\label{def:recurrence}
A point $x\in\strats$ is said to be \emph{recurrent} under $\flow$ if, for every neighborhood $U$ of $x$ in $\strats$, there exists an increasing sequence of times $t_{n}\uparrow\infty$ such that $\flow(x,t_{n})\in U$ for all $n$.
Moreover, the flow $\flow$ is called \emph{\textpar{Poincaré} recurrent} if, for every measurable subset $A$ of $\feas$, the set of recurrent points in $A$ has full measure.
\end{definition}

An immediate consequence of \cref{def:recurrence} is that, if a point is recurrent, there exists an increasing sequence of times $t_{n}\uparrow\infty$ such that $\flow(x,t_{n})\to x$.
On that account, recurrence can be seen as the flip side of convergence:
under the latter, (almost) every initial state of the dynamics eventually reaches some well-defined end-state;
instead, under the former, the system's orbits fill the entire state space and return arbitarily close to their starting points infinitely often (so there is no possibility of convergence beyond trivial cases).

\subsection{Zero-sum games with an interior equilibrium}
\label{sec:interior}

Our first result is that \eqref{eq:FRL} is recurrent (and hence, non-convergent) in zero-sum games with an interior \acl{NE}:

\begin{theorem}
\label{thm:recurrence}
Let $\game$ be a $2$-player zero-sum game that admits an interior \acl{NE}.
Then, almost every solution trajectory of \eqref{eq:FRL} is recurrent;
specifically, for \textpar{Lebesgue} almost every initial condition $x(0) = \choice(y(0)) \in\strats$, there exists an increasing sequence of times $t_{n}\uparrow\infty$ such that $x(t_{n})\to x(0)$.
\end{theorem}

The proof of \cref{thm:recurrence} is fairly complicated, so we outline the basic steps below:
\begin{enumerate}
\item
We first show that the dynamics of the score sequence $y(t)$ are \emph{incompressible}, i.e. the volume of a set of initial conditions remains invariant as the dynamics evolve over time.
By Poincaré's recurrence theorem (cf. \cref{app:dynamics}), if every solution orbit $y(t)$ of \eqref{eq:FRL} remains in a compact set for all $t\geq0$, incompressibility implies recurrence.

\item
To counter the possibility of solutions escaping to infinity, we introduce a transformed system based on the differences between scores (as opposed to the scores themselves).
%
To establish boundedness in these dynamics, we consider the ``primal-dual'' coupling
\begin{equation}
\label{eq:Fenchel}
\fench(y)
	= \sum_{\play\in\players} \bracks{h_{\play}^{\ast}(y_{\play}) - \braket{y_{\play}}{\eq_{\play}}},
\end{equation}
where $\eq$ is an interior \acl{NE} and $h_{\play}^{\ast}(y_{\play}) = \max_{x_{\play}\in\strats_{\play}} \{\braket{y_{\play}}{x_{\play}} - h_{\play}(x_{\play})\}$
denotes the convex conjugate of $h_{\play}$.%
\footnote{This coupling is closely related to the so-called \emph{Bregman divergence} \textendash\ for the details, see \cite{Kiw97b,ABRT04,SS11,MS16}.}
The key property of this coupling is that it remains invariant under \eqref{eq:FRL};
however, its level sets are not bounded so, again, precompactness of solutions is not guaranteed.

\item
Nevertheless, under the score transformation described above,
the level sets of $\fench$ \emph{are} compact.
Since the transformed dynamics are invariant under said transformation, Poincaré's theorem finally implies recurrence.
\end{enumerate}


\begin{proof}[Proof of \cref{thm:recurrence}]
To make the above plan precise, fix some ``benchmark'' strategy $\ptest_{\play}\in\pures_{\play}$ for every player $\play\in\players$ and, for all $\pure_{\play} \in \pures_{\play} \exclude{\ptest_{\play}}$, consider the corresponding score differences
\begin{equation}
\label{eq:score-z}
z_{\play\pure_{\play}}
	= y_{\play\pure_{\play}} - y_{\play,\ptest_{\play}}.%
\end{equation}
Obviously, $z_{\play\ptest_{\play}} = y_{\play\ptest_{\play}} - y_{\play\ptest_{\play}}$ is identically zero so we can ignore it in the above definition.
In so doing, we obtain a linear map $\Pi_{\play}\from\R^{\pures_{\play}} \to \R^{\pures_{\play}\exclude{\ptest_{\play}}}$ sending $y_{\play} \mapsto z_{\play}$;
aggregating over all players, we also write $\Pi$ for the product map $\Pi = (\Pi_{1},\dotsc,\Pi_{\nPlayers})$ sending $y\mapsto z$.
For posterity, note that this map is surjective but \emph{not} injective,%
\footnote{Specifically, $\Pi_{\play}(y_{\play}) = \Pi_{\play}(y_{\play}')$ if and only if $y_{\play\pure_{\play}}' = y_{\play\pure_{\play}} + c$ for some $c\in\R$ and all $\pure_{\play}\in\pures_{\play}$.}
so it does not allow us to recover the score vector $y$ from the score difference vector $z$.

Now, under \eqref{eq:FRL}, the score differences \eqref{eq:score-z} evolve as
\begin{equation}
\label{eq:dyn-z}
\dot z_{\play\pure_{\play}}
	= \payv_{\play\pure_{\play}}(x(t)) - \payv_{\play\ptest_{\play}}(x(t)).
\end{equation}
However, since the \ac{RHS} of \eqref{eq:dyn-z} depends on $x = \choice(y)$ and the mapping $y\mapsto z$ is not invertible (so $y$ cannot be expressed as a function of $z$), the above does not a priori constitute an autonomous dynamical system (as required to apply Poincaré's recurrence theorem).
Our first step below is to show that \eqref{eq:dyn-z} does in fact constitute a well-defined dynamical system on $z$.

To do so, consider the reduced choice map $\hat\choice_{\play}\from\R^{\pures_{\play}\exclude{\ptest_{\play}}} \to \strats_{\play}$ defined as
\begin{equation}
\label{eq:choice-z}
\hat\choice_{\play}(z_{\play})
	= \choice_{\play}(y_{\play})
\end{equation}
for some $y_{\play}\in\R^{\pures_{\play}}$ such that $\Pi_{\play}(y_{\play}) = z_{\play}$.
That such a $y_{\play}$ exists is a consequence of $\Pi_{\play}$ being surjective;
furthemore, that $\hat\choice_{\play}(z_{\play})$ is well-defined is a consequence of the fact that $\choice_{\play}$ is invariant on the fibers of $\Pi_{\play}$.
Indeed, by construction, we have $\Pi_{\play}(y_{\play}) = \Pi_{\play}(y_{\play}')$ if and only if $y_{\play\pure_{\play}}' = y_{\play\pure_{\play}} + c$ for some $c\in\R$ and all $\pure_{\play}\in\pures_{\play}$.
Hence, by the definition of $\choice_{\play}$, we get
\begin{flalign}
\choice_{\play}(y_{\play}')
	&\txs
	= \argmax\limits_{x_{\play}\in\strats_{\play}} \braces*{\braket{y_{\play}}{x_{\play}} + c \sum_{\pure_{\play}\in\pures_{\play}} x_{\play\pure_{\play}} - h_{\play}(x_{\play})}
	\notag\\
	&= \argmax_{x_{\play}\in\strats_{\play}} \braces{\braket{y_{\play}}{x_{\play}} - h_{\play}(x_{\play})}
	= \choice_{\play}(y_{\play}),
\end{flalign}
where we used the fact that $\sum_{\pure_{\play}\in\pures_{\play}} x_{\play\pure_{\play}} = 1$.
The above shows that $\choice_{\play}(y_{\play}') = \choice_{\play}(y_{\play})$ if and only if $\Pi_{\play}(y_{\play}) = \Pi_{\play}(y_{\play}')$, so $\hat\choice_{\play}$ is well-defined.

Letting $\hat\choice \equiv (\hat\choice_{1},\dotsc,\hat\choice_{\nPlayers})$ denote the aggregation of the players' individual choice maps $\hat\choice_{\play}$, it follows immediately that $\choice(y) = \hat\choice(\Pi(y)) = \hat\choice(z)$ by construction.
Hence, the dynamics \eqref{eq:dyn-z} may be written as
\begin{equation}
\label{eq:FRL-z}
\dot z
	= \dynfield(z),
\end{equation}
where
\begin{equation}
\dynfield_{\play\pure_{\play}}(z)
	= \payv_{\play\pure_{\play}}(\hat\choice_{\play}(z)) - \payv_{\play\ptest_{\play}}(\hat\choice_{\play}(z)).
\end{equation}
These dynamics obviously constitute an autonomous system, so our goal will be to use Liouville's formula and Poincaré's theorem in order to establish recurrence and then conclude that the induced trajectory of play $x(t)$ is recurrent by leveraging the properties of $\hat\choice$.

As a first step towards applying Liouville's formula, we note that the dynamics \eqref{eq:FRL-z} are \emph{incompressible}.
Indeed, we have
\begin{equation}
\label{eq:incompressible}
\frac{\pd\dynfield_{\play\pure_{\play}}}{\pd z_{\play\pure_{\play}}}
	= \sum_{\purealt_{\play}\in\pures_{\play}}
		\frac{\pd\dynfield_{\play\pure_{\play}}}{\pd x_{\play\purealt_{\play}}}
		\frac{\pd x_{\play\purealt_{\play}}}{\pd z_{\play\pure_{\play}}}
	= 0,
\end{equation}
because $\payv_{\play}$ does not depend on $x_{\play}$.
We thus obtain $\divg_{z} \dynfield(z) = 0$, i.e. the dynamics \eqref{eq:FRL-z} are incompressible.

We now show that every solution orbit $z(t)$ of \eqref{eq:FRL-z} is \emph{precompact}, that is, $\sup_{t\geq0} \norm{z(t)} < \infty$.
To that end, note that the coupling $\fench(y) = \sum_{\play\in\players} \bracks{h_{\play}^{\ast}(y_{\play}) - \braket{y_{\play}}{\eq_{\play}}}$ defined in \eqref{eq:Fenchel} remains invariant under \eqref{eq:FRL} when $\game$ is a $2$-player zero-sum game.
Indeed, by \cref{lem:dFench}, we have
\begin{flalign}
\label{eq:dFench}
\frac{d\fench}{dt}
	= \sum_{\play\in\players} \braket{\payv_{\play}(x)}{x_{\play} - \eq_{\play}}
	&= \braket{\payv_{1}(x)}{x_{1} - \eq_{1}}
		+ \braket{\payv_{2}(x)}{x_{2} - \eq_{2}}
	\notag\\
	&= \pay_{1}(x_{1},x_{2}) - \pay_{1}(\eq_{1},x_{2})
		+ \pay_{2}(x_{1},x_{2}) - \pay_{2}(x_{1},\eq_{2})
	= 0,
\end{flalign}
where we used the fact that $\choice_{\play} = \nabla h_{\play}^{\ast}$ in the first line (cf. \eqref{eq:choice-conj} above), and the assumption that $\eq$ is an interior \acl{NE} of a $2$-player zero-sum game in the last one.
We conclude that $\fench(y(t))$ remains constant under \eqref{eq:FRL}, as claimed.

By \cref{lem:Fenchel-bounded} in \cref{app:aux}, the invariance of $\fench(y(t))$ under \eqref{eq:FRL} implies that the score differences $z_{\play\pure_{\play}}(t) = y_{\play\pure_{\play}}(t) - y_{\play\ptest_{\play}}(t)$ also remain bounded for all $t\geq0$.
Hence, by Liouville's formula and Poincaré's recurrence theorem, the dynamics \eqref{eq:FRL-z} are \emph{recurrent}, i.e. for (Lebesgue) almost every initial condition $z_{0}$ and every neighborhood $U$ of $z_{0}$, there exists some $\tau_{U}$ such that $z(\tau_{U})\in U$ (cf. \cref{def:recurrence}).
Thus, taking a shrinking net of balls $\ball_{\run}(z_{0}) = \setdef{z}{\norm{z-z_{0}}\leq1/n}$ and iterating the above, it follows that there exists an increasing sequence of times $t_{n}\uparrow\infty$ such that $z(t_{n})\to z_{0}$.
Therefore, to prove the corresponding claim for the induced trajectories of play $x(t) = \choice(y(t)) = \hat\choice(z(t))$ of \eqref{eq:FRL}, fix an initial condition $x_{0}\in\intstrats$ and take some $z_{0}$ such that $x_{0} = \hat\choice(z_{0})$.
By taking $t_{\run}$ as above, we have $z(t_{\run})\to z_{0}$ so, by continuity, $x(t_{\run}) = \hat\choice(z_{\run}) \to \hat\choice(z_{0}) = x_{0}$.
This shows that any solution orbit $x(t)$ of \eqref{eq:FRL} is recurrent and our proof is complete.
\end{proof}

\begin{remark*}
We close this section by noting that the invariance of \eqref{eq:Fenchel} under \eqref{eq:FRL} induces a foliation of $\feas$, with each individual trajectory of \eqref{eq:FRL} living on a ``leaf'' of the foliation (a level set of $\fench$).
\cref{fig:recurrence} provides a schematic illustration of this foliation/cycling structure.
\end{remark*}

\subsection{Zero-sum games with no interior equilibria}
\label{sec:boundary}

At first sight, \cref{thm:recurrence} suggests that cycling is ubiquitous in zero-sum games;
however, if the game does not admit an interior equilibrium, the behavior of \eqref{eq:FRL} turns out to be qualitatively different.
To state our result for such games, it will be convenient to assume that the players' regularizer functions are \emph{strongly convex},
i.e. each $h_{\play}$ can be bounded from below by a quadratic minorant:
\begin{equation}
\label{eq:str-cvx}
h_{\play}(tx_{\play} + (1-t)x_{\play}')
	\leq t h_{\play}(x_{\play}) + (1-t) h_{\play}(x_{\play}') - \tfrac{1}{2} \strong_{\play} t(1-t) \norm{x_{\play} - x_{\play}'}^{2},
\end{equation}
for all $x_{\play}, x_{\play}'\in\strats_{\play}$ and for all $t\in[0,1]$.
Under this technical assumption, we have:

\begin{theorem}
\label{thm:recurrence-boundary}
Let $\game$ be a $2$-player zero-sum game that does not admit an interior \acl{NE}.
Then, for every initial condition of \eqref{eq:FRL}, the induced trajectory of play $x(t)$ converges to the boundary of $\strats$.
Specifically, if $\eq$ is a \acl{NE} of $\game$ with maximal support, $x(t)$ converges to the relative interior of the face of $\strats$ spanned by $\supp(\eq)$.
\end{theorem}

\cref{thm:recurrence-boundary} is our most comprehensive result for the behavior of \eqref{eq:FRL} in zero-sum games, so several remarks are in order.
First, we note that \cref{thm:recurrence-boundary} complements \cref{thm:recurrence} in a very natural way:
specifically, if $\game$ admits an interior \acl{NE}, \cref{thm:recurrence-boundary} suggests that the solutions of \eqref{eq:FRL} will stay within the relative interior $\intstrats$ of $\strats$ (since an interior equilibrium is supported on all actions).
Of course, \cref{thm:recurrence} provides a stronger result because it states that, within $\intstrats$, \eqref{eq:FRL} is recurrent.
Hence, applying both results in tandem, we obtain the following heuristic for the behavior of \eqref{eq:FRL} in zero-sum games:
\begin{quote}
\centering
\itshape
In the long run, \eqref{eq:FRL} wanders in perpetuity\\
in the smallest face of $\strats$ containing the equilibrium set of $\game$.
\end{quote}

This leads to two extremes:
On the one hand, if $\game$ admits an interior equilibrium, \eqref{eq:FRL} is recurrent and cycles in the level sets of the coupling function \eqref{eq:Fenchel}.
At the other end of the spectrum, if $\game$ admits only a single, pure equilibrium, then \eqref{eq:FRL} converges to it (since it has to wander in a singleton set).
In all other ``in-between'' cases, \eqref{eq:FRL} exhibits a hybrid behavior, converging to the face of $\strats$ that is spanned by the maximal support equilibrium of $\game$, and then cycling in that face in perpetuity.

The reason for this behavior is that the coupling \eqref{eq:Fenchel} is no longer a constant of motion of \eqref{eq:FRL} if the game does not admit an interior equilibrium.
As we show in \cref{app:proofs}, the coupling \eqref{eq:Fenchel} is strictly decreasing when the support of $x(t)$ is strictly greater than that of a \acl{NE} $\eq$ with maximal support.
When the two match, the rate of change of \eqref{eq:Fenchel} drops to zero, and we fall back to a ``constrained'' version of \cref{thm:recurrence}.
We make this argument precise in \cref{app:proofs} (where we present the proof of \cref{thm:recurrence-boundary}).

\subsection{Zero-sum polymatrix games \& positive affine payoff transformations}

We close this section by showing that the recurrence properties of \eqref{eq:FRL} are not unique to ``vanilla'' zero-sum games, but also occur when there is a \emph{network of competitors} \textendash\ i.e. in $\nPlayers$-player zero-sum polymatrix games.
In fact, the recurrence results carry over to any $\nPlayers$-player game which is isomorphic to a constant-sum polymatrix game with an interior equilibrium up to a positive-affine payoff transformation (possibly different transformation for each agent).
For example, this class of games contains all strictly competitive games \cite{adler2009note}.  
Such transformations do not affect the equilibrium structure of the game, but can affect  the geometry of the trajectories;
nevertheless, the recurrent behavior persists as shown by the following result:

\begin{theorem}
\label{thm:recurrence-net}
Let $\game = (\game_{\edge})_{\edge\in\edges}$ be a constant-sum polymatrix game
\textpar{or a positive affine payoff transformation thereof}.
If $\game$ admits an interior \acl{NE}, almost every solution trajectory of \eqref{eq:FRL} is recurrent;
specifically, for \textpar{Lebesgue} almost every initial condition $x(0) = \choice(y(0)) \in \strats$, there exists an increasing sequence of times $t_{n}\uparrow\infty$ such that $x(t_{n})\to x(0)$.
\end{theorem}

We leave the case of zero-sum polymatrix games with no interior equilibria to future work.

\section{Conclusions}
\label{sec:conclusion}

Our results show that the behavior of regularized learning in adversarial environments is considerably more intricate than the strong no-regret properties of \ac{FRL} might at first suggest.
Even though the empirical frequency of play under \ac{FRL} converges to the set of \aclp{CCE} (possibly at an increased rate, depending on the game's structure), the actual trajectory of play under \ac{FRL} is recurrent and exhibits cycles in zero-sum games.
We find this property particularly interesting as it suggests that ``black box'' guarantees are not the be-all/end-all of learning in games:
the theory of dynamical systems is rife with complex phenomena and notions that arise naturally when examining the behavior of learning algorithms in finer detail.


\appendix

\section{Examples of \ac{FRL} dynamics}
\label{app:examples}

\begin{example}[Multiplicative weights and the replicator dynamics]
\label{ex:choice-logit}
Perhaps the most widely known example of a regularized choice map is the so-called \emph{logit choice map}
\begin{equation}
\label{eq:choice-logit}
\logit_{\play}(y)
	= \frac{(\exp(y_{\play\pure_{\play}}))_{\pure_{\play}\in\pures_{\play}}}{\sum_{\purealt_{\play}\in\pures_{\play}} \exp(y_{\play\purealt_{\play}})}.
\end{equation}
This choice model was first studied in the context of discrete choice theory by McFadden \cite{McF74a} and it leads to the \acdef{MW} dynamics:%
\footnote{The terminology ``\acl{MW}'' refers to the fact that \eqref{eq:MW} is the continuous version of the discrete-time \acl{MWU} rule:
\begin{equation}
\label{eq:MWU}
\tag{MWU}
x_{\play\pure_{\play}}(t+1)
	= \frac
	{x_{\play\pure_{\play}}(t) e^{\temp_{\play} \payv_{\play\pure_{\play}}(x(t))}}
	{\sum_{\purealt_{\play}\in\pures_{\play}} x_{\play\purealt_{\play}}(t) e^{\temp_{\play} \payv_{\play\purealt_{\play}}(x(t))}},
\end{equation}
where $\temp_{\play}>0$ is the scheme's ``learning rate''.
For more details about \eqref{eq:MWU}, we refer the reader to \cite{Arora05themultiplicative}.}
\begin{equation}
\label{eq:MW}
\tag{MW}
\begin{aligned}
\dot y_{\play}
	&= \payv_{\play}(x),
	\\
x_{\play}
	&= \logit_{\play}(y_{\play}).
\end{aligned}	
\end{equation}
As is well known, the logit map above is obtained by the model \eqref{eq:choice} by considering the entropic regularizer
\begin{equation}
\label{eq:penalty-Gibbs}
h_{\play}(x)
	= \sum_{\pure_{\play}\in\pures_{\play}} x_{\play\pure_{\play}} \log x_{\play\pure_{\play}},
\end{equation}
i.e. the (negative) \emph{Gibbs\textendash Shannon entropy function}.
A simple differentiation of \eqref{eq:MW} then shows that the players' mixed strategies evolve according to the dynamics
\begin{equation}
\label{eq:RD}
\tag{RD}
\dot x_{\play\pure_{\play}}
	= x_{\play\pure_{\play}}
		\bracks*{\payv_{\play\pure_{\play}}(x) - \sum_{\purealt_{\play}\in\pures_{\play}} x_{\play\purealt_{\play}} \payv_{\play\purealt_{\play}}(x)},
\end{equation}
This equation describes the \emph{replicator dynamics} of \cite{TJ78}, the most widely studied model for evolution under natural selection in population biology and evolutionary game theory.
The basic relation between \eqref{eq:MW} and \eqref{eq:RD} was first noted in a single-agent environment by \cite{Rus99} and was explored further in game theory by \cite{HSV09,Sor09,MM10,MS16} and many others.
\end{example}

\begin{example}[Euclidean regularization and the projection dynamics]
\label{ex:choice-Eucl}
Another widely used example of regularization is given by the \emph{quadratic penalty}
\begin{equation}
\label{eq:penalty-Eucl}
h_{\play}(x_{\play})
	= \frac{1}{2} \sum_{\pure_{\play}\in\pures_{\play}} x_{\play\pure_{\play}}^{2}.
\end{equation}
The induced choice map \eqref{eq:choice} is the (Euclidean) \emph{projection map}
\begin{equation}
\label{eq:choice-Eucl}
\txs
\Eucl_{\play}(y_{\play})
	= \argmax_{x_{\play}\in\strats_{\play}}
		\braces[\big]{\braket{y_{\play}}{x_{\play}} - \tfrac{1}{2}\norm{x_{\play}}_{2}^{2} }
	= \argmin_{x_{\play}\in\strats_{\play}}
	\norm{y_{\play} - x_{\play}}_{2}^{2},
\end{equation}
leading to the \emph{projected reinforcement learning} process
\begin{equation}
\label{eq:PL}
\tag{PL}
\begin{aligned}
\dot y_{\play}
	&= \payv_{\play}(x),
	\\
x_{\play}
	&= \Eucl_{\play}(y_{\play}).
\end{aligned}
\end{equation}
The players' mixed strategies are then known to follow the \emph{projection dynamics}
\begin{equation}
\label{eq:PD}
\tag{PD}
\dot x_{\play\pure_{\play}}
	= \begin{cases}
	\dis\payv_{\play\pure_{\play}}(x) - \abs{\supp(x_{\play})}^{-1} \insum_{\purealt_{\play}\in\supp(x_{\play})} \payv_{\play\purealt_{\play}}(x)
		&\text{if $\pure_{\play}\in\supp(x_{\play})$},
		\\
	0
		&\text{if $\pure_{\play}\notin \supp(x_{\play})$},
	\end{cases}
\end{equation}
over all intervals for which the support of $x(t)$ remains constant \cite{MS16}.
The dynamics \eqref{eq:PD} were introduced in game theory by \cite{Fri91} as a geometric model of the evolution of play in population games;
for a closely related approach, see also \cite{NZ97,LS08} and references therein.
\end{example}

\section{Liouville's formula and Poincaré recurrence}
\label{app:dynamics}
\noindent
Below we present for completeness some basic results from the theory of dynamical systems.

\subsection*{Liouville's Formula}

Liouville's formula can be applied to any system of autonomous differential equations with a continuously differentiable
vector field $\xi$ on an open domain of $\mathcal{S} \subset \real^k$.  The divergence of $\xi$ at $x \in \mathcal{S}$ is defined
as the trace of the corresponding Jacobian at $x$, \textit{i.e.},  $\text{div}[\xi(x)]=\sum_{i=1}^k \frac{\partial \xi_i}{\partial x_i}(x)$. 
Since divergence is a continuous function we can compute its integral over measurable sets $A\subset \mathcal{S}$. Given 
any such set $A$, let $A(t)= \{\Phi(x_0,t): x_0 \in A\}$ be the image of $A$ under map $\Phi$ at time $t$. $A(t)$ is measurable
and is volume is $\text{vol}[A(t)]= \int_{A(t)}dx$. Liouville's formula states that the time derivative of the volume $A(t)$ exists and
is equal to the integral of the divergence over $A(t)$:

$$\frac{d}{dt} [A(t)] = \int_{A(t)} \text{div} [\xi(x)]dx.$$

A vector field is called divergence free if its divergence is zero everywhere. Liouville's formula trivially implies that volume is preserved
in such flows.

\subsection*{Poincar\'{e}'s recurrence theorem}
The notion of recurrence that we will be using in this paper goes back to Poincar\'{e} and specifically to his study of the three-body problem. In 1890, in his celebrated work \cite{Poincare1890}, he proved that whenever a dynamical system preserves volume almost all trajectories return arbitrarily close to their initial position, and they do so an infinite number of times. More precisely, Poincar\'{e} established the following:
\medskip

\textbf{Poincaré Recurrence:}
\textit{\cite{Poincare1890,barreira}
%
If a flow preserves volume and has only bounded orbits then for each open
set there exist orbits that intersect the set infinitely often.
}

\section{Technical proofs}
\label{app:proofs}

The first result that we prove in this appendix is a key technical lemma concerning the evolution of the coupling function \eqref{eq:Fenchel}:

\begin{lemma}
\label{lem:dFench}
Let $\base_{\play}\in\strats_{\play}$ and let $\fench_{\play}(y_{\play}) = h_{\play}^{\ast}(y_{\play}) - \braket{y_{\play}}{\base_{\play}}$ denote the coupling \eqref{eq:Fenchel} for player $\play\in\players$.
If player $\play\in\players$ follows \eqref{eq:FRL}, we have
\begin{equation}
\label{eq:dFench}
\frac{d}{dt} \fench_{\play}(y_{\play}(t))
	= \braket{\payv_{\play}(x(t))}{x_{\play}(t) - \base_{\play}},
\end{equation}
for every trajectory of play $x_{-\play}(t)$ of all players other than $\play$.
\end{lemma}

\begin{proof}
We begin by recalling the ``maximizing argument'' identity
\begin{equation}
\label{eq:choice-conj}
\choice_{\play}(y_{\play}) = \nabla h_{\play}^{\ast}(y_{\play})
\end{equation}
which expresses the choice map $\choice_{\play}$ as a function of the convex conjugate of $h_{\play}$ \cite[p.~149]{SS11}.
With this at hand, a simple differentiation gives
\begin{flalign}
\frac{d}{dt} \fench_{\play}(y_{\play}(t))
	&= \frac{d}{dt} h_{\play}^{\ast}(y_{\play}(t)) - \braket{\dot y_{\play}(t)}{\base_{\play}}
	\notag\\
	&= \braket{\dot y_{\play}(t)}{\nabla h_{\play}^{\ast}(y_{\play}(t)) - \base_{\play}}
	\notag\\
	&= \braket{\payv_{\play}(x(t))}{x_{\play}(t) - \base_{\play}},
\end{flalign}
where the last step follows from the fact that $x_{\play}(t) = \choice_{\play}(y_{\play}(t)) = \nabla h_{\play}^{\ast}(y_{\play}(t))$.
\end{proof}

Armed with this lemma, we proceed to prove the no-regret guarantees of \eqref{eq:FRL}:

\begin{proof}[Proof of \cref{thm:no-reg}]
Fix some base point $\base_{\play}\in\strats_{\play}$ and let $\lyap_{\play}(t) = \fench_{\play}(y_{\play}(t)) = h_{\play}(y_{\play}(t)) - \braket{y_{\play}(t)}{\base_{\play}}$.
Then, by \cref{lem:dFench}, we have
\begin{equation}
\lyap_{\play}'(t)
	= \braket{\payv_{\play}(x(t))}{x_{\play}(t) - \base_{\play}}
\end{equation}
and hence, after integrating and rearranging, we get
\begin{equation}
\label{eq:no-reg1}
\int_{0}^{t} \bracks{\pay_{\play}(\base_{\play};x_{-\play}(s)) - \pay_{\play}(x(s))} \dd s
	= \int_{0}^{t} \braket{\payv_{\play}(x(s))}{\base_{\play} - x_{\play}(s)} \dd s
	= \lyap_{\play}(0) - \lyap_{\play}(t),
\end{equation}
where we used the fact that $\pay_{\play}(\base_{\play};x_{-\play}) = \braket{\payv_{\play}(x)}{\base_{\play}}$ \textendash\ cf.~\cref{eq:pay-payv} in \cref{sec:prelims}.
However, expanding the \acs{RHS} of \eqref{eq:no-reg1}, we get
\begin{flalign}
\lyap_{\play}(0) - \lyap_{\play}(t)
	&= h_{\play}^{\ast}(y_{\play}(0)) - \braket{y_{\play}(0)}{\base_{\play}}
	- h_{\play}^{\ast}(y_{\play}(t)) + \braket{y_{\play}(t)}{\base_{\play}}
	\notag\\
	&\leq h_{\play}^{\ast}(y_{\play}(0)) - \braket{y_{\play}(0)}{\base_{\play}} + h_{\play}(\base_{\play})
	\notag\\
	&= h_{\play}(\base_{\play}) - h_{\play}(\choice_{\play}(y_{\play}(0)))
	\notag\\
	&\leq \max h_{\play} - \min h_{\play}
	\equiv \depth_{\play},
\end{flalign}
where we used the defining property of convex conjugation in the second and third lines above \textendash\ i.e. that $h_{\play}^{\ast}(y_{\play}) \geq \braket{y_{\play}}{x_{\play}} - h_{\play}(x_{\play})$ for all $x_{\play}\in\strats_{\play}$, with equality if and only if $x_{\play} = \choice_{\play}(y_{\play})$.
Thus, maximizing \eqref{eq:no-reg1} over $\base_{\play}\in\strats_{\play}$, we finally obtain
\begin{equation}
\label{eq:no-reg2}
\reg_{\play}(t)
	= \max_{\base_{\play}\in\strats_{\play}}
		\frac{1}{t}\int_{0}^{t} \bracks{\pay_{\play}(\base_{\play};x_{-\play}(s)) - \pay_{\play}(x(s))} \dd s
		\leq \frac{\depth_{\play}}{t},
\end{equation}
as claimed.
\end{proof}

We now turn to two-player zero-sum games that do not admit interior equilibria.
To describe such equilibria in more detail, we consider below the notion of \emph{essential} and \emph{non-essential} strategies:

\begin{definition}
\label{def:essential}
A strategy $\pure_{\play}$ of agent $\play \in \{1,2\}$ in a zero sum game is called \emph{essential} if there exists a \acl{NE} in which player $\play$ plays $\pure_{\play}$ with positive probability.
A strategy that is not essential is called \emph{non-essential}.
\end{definition}

As it turns out, the \aclp{NE} of a zero-sum game admit a very useful characterization in terms of essential and non-essential strategies:

\begin{lemma}
\label{lem:Farka_eq}
Let $\game$ be a $2$-player zero-sum game that does not admit an interior \acl{NE}.
Then, there exists a mixed \acl{NE} $(\strat_{1}, \strat_{2})$ such that
\begin{inparaenum}
[\itshape a\upshape)]
\item
each agent plays each of their essential strategies with positive probability;
and
\item
for each agent deviating to a non-essential strategy results to a strictly worse performance than the value of the game.
\end{inparaenum}
\end{lemma}

The key step in proving this characterization is Farkas' lemma;
the version we employ here is due to Gale, Kuhn and Tucker \cite{BookGKT}):
\PM{Provide reference? Done.}

\begin{lemma}[Farkas' lemma]
Let $\bP \in \R^{m\times n}$  and $\boldb \in \R^{m}$. Then exactly one of the following two statements is true:
 \begin{itemize}
\item There exists a $\bx \in \R^{m}$ such that $\bP^{\top}\bx \geq 0$ and $\boldb^{\top}\bx  < 0$.
\item There exists a $\by \in \R^{n}$ such that $\bP\cdot\by = \boldb$ and $\by \geq 0$.
\end{itemize}
\end{lemma}

With this lemma at hand, we have:

\PM{Some final fine-tuning needed here, will do so tomorrow.}

\begin{proof}[Proof of \cref{lem:Farka_eq}]
Assume without loss of generality that the value of the zero-sum game is zero. 
and that the first agent is a maximizing agent.
Let $A$ be the payoff matrix of the first agent and hence $A^{T}=A$ the payoff matrix of the second/minimizing agent. We will show first that for any non-essential strategy $\pure_\play$ of each agent there exists a \acl{NE} strategy of his opponent such that the expected performance of $\pure_\play$ is strictly worse than the value of the game (i.e. zero).   

It suffices to argue this for the first agent. Let  $\pure_\play$ be one of his non-essential strategies then by definition there does not exist any \acl{NE} strategy of that agent that chooses  $\pure_\play$ with positive probability. This is equivalent to the negation of the following statement:

There exists a $\bx \in \R^{m}$ such that $\mathbf{P}^{\top}\bx \geq 0$ and $\mathbf{b}^{\top}\bx  < 0$

\noindent
where 

\begin{equation} \label{first_condition}
 \mathbf{P}^{\top} =\begin{pmatrix}                                            
                                                   \mathbf{A}^{\top} \\ 
                                                    \mathbf{I}_{m\times m}                                               
                                                    \end{pmatrix} = \begin{pmatrix}
                                  a_{11}& a_{21} & \ldots & a_{m1}\\
                                   \vdots & \vdots & \ldots & \vdots\\
                                  a_{1n}& a_{2n} & \ldots & a_{mn}\\
                                    1 & 0 & \ldots & 0 \\
                                  0&  1& \ldots &   0\\
                                  \vdots & \vdots & \ldots & \vdots\\
                                  0&  0& \ldots &   1\\ 
                                \end{pmatrix} ,
\end{equation}

\noindent
and
$\mathbf{b}= -\mathbf{e}_{i} = (0, \dots, 0, -1, 0, \dots, 0)^T$, the standard basis vector of dimension $m$ that ``chooses" the $\play$-th strategy.  By Farkas' lemma,  there exists a $\mathbf{y} \in \R^{m+n}$ such that $\mathbf{Py} = \mathbf{b}$ and $\mathbf{y} \geq 0$. It is convenient to express $\mathbf{y}= (\mathbf{z}; \mathbf{w})$ where    $\mathbf{z} \in  \R^{n}$ and $\mathbf{w} \in  \R^{m}$.
Hence, for all $j \neq \play \in \{1,2, \dots, m\}: \mathbf{(Py)}_j= \mathbf{(Az)}_j+\mathbf{w}_j=0$ and thus $\mathbf{(Az)}_j\leq 0$. Finally, for $j=\play: \mathbf{(Py)}_\play= \mathbf{(Az)}_\play+\mathbf{w}_\play=-1$ and thus $\mathbf{(Az)}_\play<0$. Hence $\mathbf{z}$ is a \acl{NE} strategy for the second player such that when the first agent chooses the non-essential strategy $\pure_\play$ he receives payoff which is strictly worse than his value (zero).

To complete the proof, for each essential strategy of the first agent there exists one equilibrium strategy of his that chooses it with positive probability (by definition). Similarly, for each non-essential strategy of the second agent there exists one equilibrium strategy of the first agent such that makes the expected payoff of that non-essential strategy strictly worse than the value of the game. The barycenter of all the above equilibrium strategies is still an equilibrium strategy (by convexity) and has all the desired properties.
\end{proof}

With all this at hand, we are finally in a position to prove \cref{thm:recurrence-boundary}:

\begin{proof}[Proof of \cref{thm:recurrence-boundary}]
We first show that the coupling $\fench(y) = \sum_{\play\in\players} \bracks{h_{\play}^{\ast}(y_{\play}) - \braket{y_{\play}}{\eq_{\play}}}$ defined in \eqref{eq:Fenchel} given any fully mixed initial condition strictly increases under \eqref{eq:FRL} when $\game$ is a $2$-player zero-sum game that does not have an equilibrium with full support. 

Indeed, by \eqref{lem:Farka_eq} there exists a mixed \acl{NE} $(\eq_{1}, \eq_{2})$ such that
\begin{inparaenum}
[\itshape i\upshape)]
\item
both players employ each of their essential strategies with positive probability over time;
and
\item
every player deviating to a non-essential strategy obtains a payoff lower than the value of the game.
\end{inparaenum}
As a result, any player playing an interior (fully mixed) strategy against such an equilibrium strategy must receive less utility than their value.
In more detail, we have
\begin{flalign}
\label{eq:dFench2}
\frac{d\fench}{dt}
	&= \sum_{\play\in\players} \braket{\payv_{\play}(x)}{x_{\play} - \eq_{\play}}
	= \braket{\payv_{1}(x)}{x_{1} - \eq_{1}}
		+ \braket{\payv_{2}(x)}{x_{2} - \eq_{2}}
	\notag\\
	&= \pay_{1}(x_{1},x_{2}) - \pay_{1}(\eq_{1},x_{2})
		+ \pay_{2}(x_{1},x_{2}) - \pay_{2}(x_{1},\eq_{2})
	\notag\\
	&= - \pay_{1}(\eq_{1},x_{2})
		 - \pay_{2}(x_{1},\eq_{2})
	\notag\\
	&< - \pay_{1}(\eq_{1},\eq_{2})
		 - \pay_{2}(\eq_{1},\eq_{2})
	= 0,
\end{flalign}
where we used the fact that $\choice_{\play} = \nabla h_{\play}^{\ast}$ in the first line (cf. \cref{app:aux}), and the assumption that $\eq$ is a \acl{NE} of a $2$-player zero-sum game such that  any agent playing an interior (fully mixed) strategy against such an equilibrium strategy must receive less utility than their value (and hence the agent himself receives more utility than the value of the game).
We thus conclude that $\fench(y(t))$ strictly increases under \eqref{eq:FRL}, as claimed.

Let $\eq=(\eq_{1},\eq_{2})$ be the \acl{NE} identified in \eqref{eq:dFench2} and let $\lyap(t) = \fench(y(t)) = \sum_{\play\in\players} \bracks{h_{\play}^{\ast}(y_{\play}(t)) - \braket{y_{\play}(t)}{\eq_{\play}}}$ denote the primal-dual coupling \eqref{eq:Fenchel} between $y(t)$ and $\eq$.
From \eqref{eq:dFench2}, we have that starting from any fully mixed strategy profile $x(0)\in \prod_{\play} \text{int}(\strats_{\play})$ and for all $t\geq 0$, $\lyap'(t) = \braket{\dot y(t)}{\nabla\fench(y(t))}<0$.
However, $\fench$ is bounded from below by $-\sum_{\play} \max_{x_{\play}\in\strats_{\play}} h_{\play}(x_{\play})$, and since $\fench(y(t))$ is strictly decreasing, it must exhibit a finite limit.

We begin by noting that $x(t) = \choice(y(t))$ is Lipschitz continuous in $t$.
Indeed, $\payv$ is Lipschitz continuous on $\strats$ by linearity;
furthermore, since the regularizer functions $h_{\play}$ are assumed $\strong_{\play}$-strongly convex, it follows that $\choice_{\play}$ is $(1/\strong_{\play})$-continuous by standard convex analysis arguments \cite[Theorem~12.60]{RW98}.
In turn, this implies that the field of motion $\dynfield(y) \equiv \payv(\choice(y))$ of \eqref{eq:FRL} is Lipschitz continuous, so the dynamics are well-posed and $y(t)$ is differentiable.
Since $\dot y = \payv$ and, in addition, $\payv$ is bounded on $\strats$, we conclude that $\dot y$ is bounded so, in particular, $y(t)$ is Lipschitz continuous on $[0,\infty)$.
We thus conclude that $x(t) = \choice(y(t))$ is Lipschitz continuous as the composition of Lipschitz continuous functions.

We now further claim that $\lyap'(t)$ is also Lipschitz continuous in $t$.
Indeed, by \eqref{eq:dFench}, we have $\lyap'(t) = \sum_{\play\in\players} \braket{\payv_{\play}(x(t))}{x_{\play}(t) - \eq_{\play}}$;
since $\payv_{\play}$ is Lipschitz continuous in $x$ and $x(t)$ is Lipschitz continuous in $t$, our claim follows trivially.
Hence, by \cref{lem:derivative}, we conclude that $\lim_{t\to\infty} \lyap'(t) = \lim_{\to\infty} \sum_{\play\in\players} \braket{\payv_{\play}(x(t)}{x_{\play}(t) - \eq_{\play}} = 0$.

By \eqref{eq:dFench2}, we know that $\lyap'(t) < 0$ as long as $x(t)$ is interior.
Hence, any $\omega$-limit $\olim$ of $x(t)$ cannot be interior (given that the embedded game does not have any interior \aclp{NE}).
Moreover, we can repeat this argument for any subspace such that the restriction of the game on that subspace (when ignoring the strategies that are played with probability zero) does not have a fully mixed NE. 
We thus conclude that the support of $\olim$ must be a subset of the support of $\eq$.
Since $\game$ does not admit an interior equilibrium, $\eq$ does not have full support, so every $\omega$-limit of $x(t)$ lies on the boundary of $\strats$, as claimed.
\end{proof}


\begin{figure}[tbp]
\centering
\subfigure{\includegraphics[width=.45\textwidth]{Figures/Cycles-Logit.pdf}}
\hfill
\subfigure{\includegraphics[width=.45\textwidth]{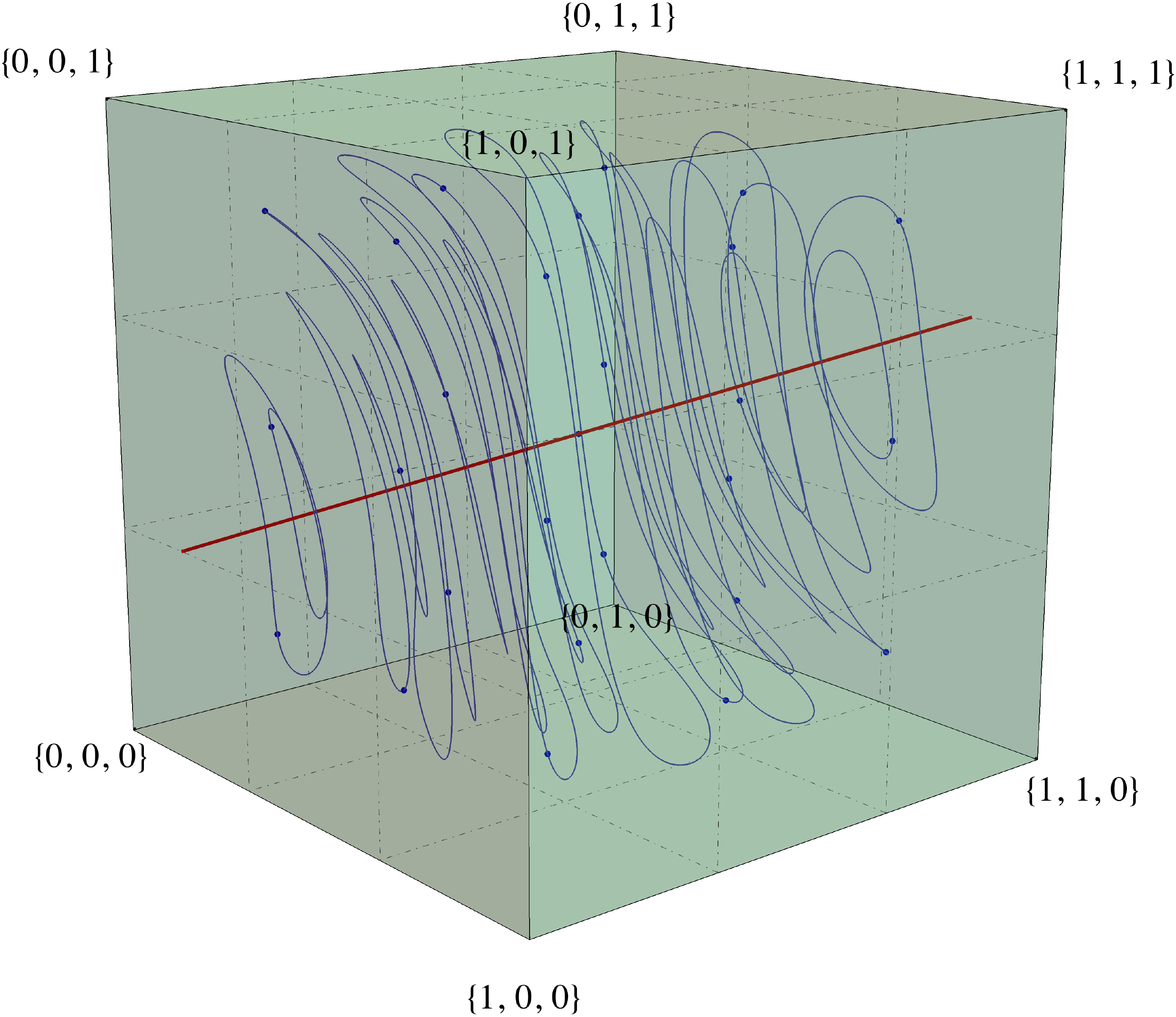}}%
\caption{Evolution of the \acl{MW} dynamics \eqref{eq:MW} in a $3$-player zero-sum polymatrix game.
In the left subfigure, each pair of players faces off in a game of standard (symmetric) Matching Pennies;
in the right, the game on each pair is weighted by a different factor.
In both cases, we plot the solution trajectories of \eqref{eq:MW} for the same initial conditions.
Even though the different weights change the trajectories of \eqref{eq:MW} and the game's equilibrium set, the cycling behavior of the dynamics remains unaffected.}%
\label{fig:recurrence-affine}
\end{figure}


We close this appendix with the proof of our result on constant-sum polymatrix games (and positive affine transformations thereof):

\begin{proof}[Proof of \cref{thm:recurrence-net}]
Our proof follows closely that of \cref{thm:recurrence};
to streamline our presentation, we only highlight here the points that differ due to working with (an positive-affine transformations of) a network of constant-sum games (as opposed to a single $2$-player zero-sum game).

The first such point is the incompressibility of the ``reduced'' dynamics \eqref{eq:FRL-z}.
By definition, we have $\pay_{\play}(x) = \sum_{\playalt\in\nodes_{\play}} \pay_{\play\playalt}(x_{\play},x_{\playalt})$, so we also have
\begin{equation}
\payv_{\play\pure_{\play}}(x)
	= \sum_{\playalt\in\nodes_{\play}} \pay_{\play\playalt}(\pure_{\play},x_{\playalt}).
\end{equation}
Since $\pay_{\play\playalt}(\pure_{\play},x_{\playalt})$ does not depend on $x_{\play}$, we readily obtain $\pd_{\pure_{\play}} \payv_{\play\pure_{\play}}(x) = 0$ and incompressibility follows as before.

Let the network game in question be isomorphic to a network of constant-sum games after the following positive-affine transformation of utilities, $\pay_{\play}(x)\leftarrow a_i \pay_{\play}(x)+b_i$ where $a_i>0$.
The second point of interest is the use of the coupling $\fench(y) = \sum_{\play\in\players} a_i \bracks{h_{\play}^{\ast}(y_{\play}) - \braket{y_{\play}}{\sol_{\play}}}$ as a constant of motion for \eqref{eq:FRL}.
Indeed, adapting the derivation of \eqref{eq:dFench}, we now get
\begin{flalign}
\label{eq:dFench-net}
\frac{d\fench}{dt}
	&= \braket{\dot y}{\nabla\fench(y)}
	= \sum_{\play\in\players} \braket{\payv_{\play}(x)}{a_i(\nabla h_{\play}^{\ast}(y_{\play}) - \eq_{\play})}
	= \sum_{\play\in\players} \braket{a_i\payv_{\play}(x)}{x_{\play} - \eq_{\play}}
	\notag\\
	&= \sum_{\play\in\players} \sum_{\playalt\in\nodes_{\play}} \braket{a_i\payv_{\play\playalt}(x)}{x_{\play} - \eq_{\play}}
	\notag\\
	&= \sum_{\{\play,\playalt\}\in\edges}
		\bracks{
			a_i\pay_{\play\playalt}(x_{\play},x_{\playalt})+b_i - a_i\pay_{\play\playalt}(\eq_{\play},x_{\playalt})-b_i
			+ a_j\pay_{\playalt\play}(x_{\play},x_{\playalt})+b_j - a_j\pay_{\playalt\play}(x_{\play},\eq_{\playalt})-b_j}
	\notag\\
	&= 0,
\end{flalign}
where the third line follows by regrouping the summands in the second line by edge, and the last line follows as in the case of \eqref{eq:dFench}.
This implies that $\fench(y(t))$ remains constant along any solution of \eqref{eq:FRL}, so the rest of the proof follows as in the case of \cref{thm:recurrence}.
\end{proof}

\section{Auxiliary results}
\label{app:aux}

In this appendix, we provide two auxiliary results that are used in the proof of \cref{thm:recurrence}.
The first one shows that if the score difference between two strategies grows large, the strategy with the lower score becomes extinct:

\begin{lemma}
\label{lem:scorediff}
Let $\pures$ be a finite set and let $h$ be a regularizer on $\strats \equiv \simplex(\pures)$.
If the sequence $y_{\run}\in\R^{\pures}$ is such that $y_{\purealt,\run} - y_{\pure,\run} \to\infty$ for some $\pure,\purealt\in\pures$, then $\lim_{\run\to\infty} \choice_{\pure}(y_{\run}) = 0$.
\end{lemma}

\begin{proof}
Set $x_{\run} = \choice(y_{\run})$ and, by descending to a subsequence if necessary, assume there exists some $\eps>0$ such that $x_{\pure,\run} \geq \eps > 0$ for all $\run$.
Then, by the defining relation $\choice(y) = \argmax\{\braket{y}{x} - h(x)\}$ of $\choice$, we have:
\begin{flalign}
\label{eq:Qcomp1}
\braket{y_{\run}}{x_{\run}} - h(x_{\run})
	\geq \braket{y_{\run}}{x'} - h(x')
\end{flalign}
for all $x'\in\simplex$.
Therefore, taking $x_{\run}' = x_{\run} + \eps(\bvec_{\purealt} - \bvec_{\pure})$, we readily obtain
\begin{equation}
\label{eq:Qcomp2}
\eps (y_{\pure,\run} - y_{\purealt,\run})
	\geq h(x_{\run}) - h(x_{\run}')
	\geq \min h - \max h
\end{equation}
which contradicts our original assumption that $y_{\pure,\run} - y_{\purealt,\run} \to -\infty$.
With $\simplex$ compact, the above shows that $x_{\pure}^{\ast} = 0$ for any limit point $x^{\ast}$ of $x_{\run}$, i.e. $\choice_{\pure}(y_{\run})\to0$.
\end{proof}

A key step of the proof of \cref{thm:recurrence} consists of showing that the level sets of the Fenchel coupling $\fench(\base,y)$ become bounded under the coordinate reduction transformation $y\mapsto\Pi(y) = z$, so every solution orbit $z(t)$ of \eqref{eq:FRL-z} also remains bounded.
We encode this in the following lemma:

\begin{lemma}
\label{lem:Fenchel-bounded}
Let $\pures$ be a finite set, let $h$ be a regularizer on $\strats \equiv \simplex(\pures)$, and fix some interior $\base\in\strats$.
If the sequence $y_{\run}\in\R^{\pures}$ is such that $\sup_{\run} \abs{h^{\ast}(y_{\run}) - \braket{y_{\run}}{\base}} < \infty$, the differences $y_{\purealt,\run} - y_{\pure,\run}$ also remain bounded for all $\pure,\purealt\in\pures$.
\end{lemma}

\begin{proof}
We argue by contradiction.
Indeed, assume that the sequence $\fench_{\run} \equiv h^{\ast}(y_{\run}) - \braket{y_{\run}}{\base}$ is bounded but $\limsup_{\run\to\infty} \abs{y_{\pure,\run} - y_{\purealt,\run}} = \infty$ for some $\pure,\purealt\in\pures$.
Letting $y_{\run}^{+} = \max_{\pure} y_{\pure,\run}$ and $y_{\run}^{-} = \min_{\pure\in\pures} y_{\pure,\run}$, this implies that $\limsup_{\run\to\infty} (y_{\run}^{+} - y_{\run}^{-}) = \infty$.
Hence, by descending to a subsequence if necessary, there exist $\pure^{+},\pure^{-} \in \pures$ such that
\begin{inparaenum}
[\itshape a\upshape)]
\item
$y_{\run}^{\pm} = y_{\pure^{\pm},\run}$ for all $\run$;
and
\item
$y_{\pure^{+},\run} - y_{\pure^{-},\run}\to\infty$ as $\run\to\infty$.
\end{inparaenum}

By construction, we have $y_{\pure^{-},\run} = y_{\run}^{-} \leq y_{\pure,\run} \leq y_{\run}^{+} = y_{\pure^{+},\run}$ for all $\pure\in\pures$.
Thus, by descending to a further subsequence if necessary, we may assume that the index set $\pures$ can be partitioned into two nonempty sets $\pures^{+}$ and $\pures^{-}$ such that
\begin{enumerate}
\item
$y_{\run}^{+} - y_{\pure,\run}$ is bounded for all $\pure\in\pures^{+}$.
\item
$y_{\run}^{+} - y_{\pure,\run}\to\infty$ for all $\pure\in\pures^{-}$.
\end{enumerate}
In more detail, consider the quantity
\begin{equation}
\delta_{\pure}
	= \liminf_{\run\to\infty} (y_{\run}^{+} - y_{\pure,\run}),
\end{equation}
and construct the required partition $\{\pures^{ +},\pures^{-}\}$ according to the following procedure:

\begin{algorithm}
\begin{algorithmic}[1]
\makeatletter
\setcounter{ALC@line}{-1}
\makeatother

\STATE
	Set $\pures^{+} \leftarrow \{\pure^{+}\}$, $\pures^{-} = \pures\setminus\pures^{+}$

\WHILE{$\delta_{\pure}<\infty$ for some $\pure\in\pures^{-}$}
	\STATE
		pick $\peq$ such that $\delta_{\peq}<\infty$;
	\STATE
		set $\pures^{+} \leftarrow \pures^{+}\cup\{\peq\}$, $\pures^{-}\leftarrow \pures^{-}\exclude{\peq}$;
	\STATE
		descend to a subsequence of $y_{\run}$ that realizes $\delta_{\peq}$;
	\STATE
		redefine $\delta_{\pure}$ for all $\pure\in\pures$ based on chosen subsequence;
\ENDWHILE

\RETURN
	$\pures^{+},\pures^{-}$
\end{algorithmic}
\end{algorithm}

Thus, if we let $x_{\run} = \choice(y_{\run})$ we readily obtain:
\begin{flalign}
\braket{y_{\run}}{\base - x_{\run}}
	&= \sum_{\pure\in\pures} y_{\pure,\run} (\base_{\pure} - x_{\pure,\run})
	= \sum_{\pure\in\pures} (y_{\pure,\run} - y_{\run}^{+}) (\base_{\pure} - x_{\pure,\run})
	\notag\\[1ex]
	&= \sum_{\pure\in\pures^{+}} (y_{\pure,\run} - y_{\run}^{+}) (\base_{\pure} - x_{\pure,\run})
	+ \sum_{\pure\in\pures^{-}} (y_{\pure,\run} - y_{\run}^{+}) (\base_{\pure} - x_{\pure,\run}),
\end{flalign}
where we used the fact that $\sum_{\pure\in\pures} \base_{\pure} = \sum_{\pure\in\pures} x_{\pure,\run} = 1$ in the first line.
The first sum above is bounded by assumption.
As for the second one, the fact that $y_{\pure^{+},\run} - y_{\pure,\run} = y_{\run}^{+} - y_{\pure,\run} \to \infty$ implies that 
$x_{\pure,\run}\to0$ for all $\pure\in\pures^{-}$ (by \cref{lem:scorediff} above).
We thus get $\liminf_{\run} (\base_{\pure} - x_{\pure,\run}) > 0$ (recall that $p\in\intsimplex$), and hence, $\sum_{\pure\in\pures^{-}} (y_{\pure,\run} - y_{\run}^{+}) (p_{\pure} - x_{\pure,\run}) \to -\infty$.

From the above, we conclude that $\braket{y_{\run}}{\base - x_{\run}} \to -\infty$ as $\run\to\infty$.
However, by construction, we also have
\begin{equation}
\fench_{\run}
	= h^{\ast}(y_{\run}) - \braket{y_{\run}}{\eq}
	= \braket{y_{\run}}{x_{\run}} - h(x_{\run}) - \braket{y_{\run}}{\eq}
	= \braket{y_{\run}}{\base - x_{\run}} - h(x_{\run}).
\end{equation}
Since $h$ is finite on $\strat$, it follows that $\fench_{\run}\to-\infty$, contradicting our assumption that $\fench_{\run}$ is bounded.
Retracing our steps, this implies that $\sup_{\run} \abs{y_{\pure,\run} - y_{\purealt,\run}} < \infty$, as claimed.
\end{proof}

The final result we state here is a technical result regarding the asymptotic behavior of the derivative of functions with a finite limit at infinity:

\begin{lemma}
\label{lem:derivative}
Suppose that $\lyap\from[0,\infty)\to\R$ is differentiable with Lipschitz continuous derivative.
If $\lim_{t\to\infty} \lyap(t)$ exists and is finite, we have $\lim_{t\to\infty} \lyap'(t) = 0$.
\end{lemma}

\begin{proof}
Assume ad absurdum that $\lim_{t\to\infty} \lyap'(t) \neq 0$.
Then, without loss of generality, we may assume there exists some $\eps>0$ and an increasing sequence $t_{n}\uparrow\infty$ such that $\lyap'(t_{n}) \geq \eps$ for all $n\in\N$.
Thus, if $M$ denotes the Lipschitz constant of $\lyap'$ and $t\in[t_{n},t_{n} + \eps/(2M)]$, we readily obtain
\begin{equation}
\abs{\lyap'(t) - \lyap'(t_{n})}
	\leq M \abs{t - t_{n}}
	\leq M \cdot \frac{\eps}{2M}
	= \frac{\eps}{2}
\end{equation}
by the Lipschitz continuity of $\lyap'$.
Since $\lyap'(t_{n}) \geq \eps$ by assumption, we conclude that $\lyap'(t) \geq \eps/2$ for all $t\in[t_{n},t_{n}+\eps/(2M)]$.
Hence, by integrating, we get $\lyap(t_{n}+\eps/(2M)) \geq \lyap(t_{n}) + \eps/(2M) \cdot (\eps/2) = \lyap(t_{n}) + \eps^{2}/(4M)$ for all $n\in\N$.
Taking $n\to\infty$ and recalling that $\lyap_{\infty} \equiv \lim_{t\to\infty} \lyap(t)$ exists and is finite, we get $\lyap_{\infty} = \lyap_{\infty} + \eps^{2}/(4M) > \lyap_{\infty}$, a contradiction.
\end{proof}

\bibliographystyle{siam}
\bibliography{IEEEabrv,../Bibliography,refer}

\end{document}